\documentclass[11pt]{article}

\usepackage{amsmath,amssymb}
\usepackage{fullpage}
\usepackage{amsthm}
\usepackage{pdfpages}
\usepackage{mathrsfs}
\usepackage{comment}
\usepackage{times}

\theoremstyle{plain}
\newtheorem{theorem}{Theorem}
\newtheorem{lemma}[theorem]{Lemma}

\newtheorem{corollary}[theorem]{Corollary}

\newenvironment{proofof}[1]{ {\noindent \em Proof of #1.}\/}{\hfill\qedsymbol\bigskip}
\newenvironment{proofsketch}{ {\noindent \em Proof sketch.}\/}{\hfill\qedsymbol\bigskip}

\newcommand{\remove}[1]{}
\newcommand{\suppress}[1]{}

\newcommand{\NN}{{\mathbb{N}}}
\newcommand{\RR}{{\mathbb{R}}}

\DeclareMathOperator{\detr}{det}
\DeclareMathOperator{\rand}{rand}

\begin{document}

\title{Parametrized Metrical Task Systems}

\author{S\'ebastien Bubeck\thanks{Microsoft Research, Redmond, United States, {\tt sebubeck@microsoft.com}} \and
Yuval Rabani\thanks{The Rachel and Selim Benin School of Computer
Science and Engineering, The Hebrew University of Jerusalem,
Jerusalem 9190416, Israel, {\tt yrabani@cs.huji.ac.il}.
Research supported in part by ISF grants 956-15 and 2553-17 and by
BSF grant 2012333. Part of this work was done while visiting Microsoft
Research at Redmond.}}

\maketitle

\begin{abstract}
We consider parametrized versions of metrical task systems and metrical service systems, two
fundamental models of online computing, where the constrained parameter is the number of 
possible distinct requests $m$. Such parametrization occurs naturally in a wide range of
applications. Striking examples are certain power management problems, which are modeled 
as metrical task systems with $m=2$. We characterize the competitive ratio in terms of the parameter
$m$ for both deterministic and randomized algorithms on hierarchically separated trees.
Our findings uncover a rich and unexpected picture that differs substantially from what is
known or conjectured about the unparametrized versions of these problems. For metrical
task systems, we show that deterministic algorithms do not exhibit any asymptotic gain 
beyond one-level trees (namely, uniform metric spaces), whereas randomized algorithms 
do not exhibit any asymptotic gain even for one-level trees. In contrast, the special case
of metrical service systems (subset chasing) behaves very differently. Both deterministic
and randomized algorithms exhibit gain, for $m$ sufficiently small compared to $n$, 
for any number of levels. Most significantly, they exhibit a large gain for uniform metric spaces 
and a smaller gain for two-level trees. Moreover, it turns out that in these cases (as well as
in the case of metrical task systems for uniform metric spaces with $m$ being an absolute 
constant), deterministic algorithms are essentially as powerful as randomized algorithms. 
This is surprising and runs counter to the ubiquitous intuition/conjecture that, for most
problems that can be modeled as metrical task systems, the randomized competitive ratio 
is polylogarithmic in the deterministic competitive ratio.
\end{abstract}

\thispagestyle{empty}
\newpage
\setcounter{page}{1}

\section{Introduction}

\paragraph{Motivation.}
Metrical task systems are a general framework proposed in the early
days of online computing to model a wide range of problems. In this
model, an algorithm controls its state among $n$ possible states. It
is presented with a sequence of requests. Each request indicates the
cost of serving it in each of the $n$ states. The algorithm must choose
a state to serve the request, then pay both the transition cost and the
service cost. It competes against the optimal offline solution, and is
$c$-competitive iff for any request sequence it pays at most $c$ times
this yardstick (plus an allowed additive term independent of the request
sequence). The transition costs are assumed to be distances in a finite
metric space whose points are the states. The seminal paper~\cite{BLS92}
introduces this model and also solves completely a fundamental problem 
regarding the proposed model. It gives a tight $2n-1$ upper and lower
bound on the competitive ratio of any deterministic algorithm. The paper
also raises the same question regarding randomized algorithms. For a
uniform metric space, it gives a lower bound of ${\cal H}_n$ and a nearly
tight upper bound of $2{\cal H}_n$, where 
${\cal H}_n = 1+\frac 1 2+\cdots+\frac 1 n$ is the $n$-th harmonic number.
In light of this result, a natural folklore conjecture asserts that in any metric
space the randomized competitive ratio is $\Theta(\log n)$. This conjecture
turned out to be far more elusive than the deterministic case. Currently, the
best universal bounds are a lower bound of 
$\Omega(\log n/\log\log n)$~\cite{BBM01,BLMN03} and a very recent upper 
bound of $O(\log^2 n)$~\cite{BCLL18}.
\newline

Many online problems can be modeled as special cases of metrical task
systems. For instance, paging~\cite{ST85} and more generally the 
$k$-server problem~\cite{MMS90} are task systems on the state space
of $k$-tuples of points in an underlying metric space. But the best bounds
that are known for these problems are much lower than the general bounds
that apply to any metrical task system. Clearly, there are constraints on the
allowed requests. For instance, in the $k$-server problem in an $n$-point
metric space, the state space has size ${n\choose k}$, but the number of
possible requests is just $n$, and they have only $0$ or $\infty$ values.
Such constraints often characterize concrete applications of metrical task
systems (e.g., various problems in networks, see~\cite{BCI97} and the
references therein). A striking 
example of the gain from instances with a restricted set of requests is the 
power management problem~\cite{ISG03}. It allows only two possible requests, 
and admits an $O(1)$-competitive algorithm~\cite{AIS08}. In view of these 
examples, it is natural to seek a meaningful measure of instance complexity 
that will provide at least partial prediction of the better bounds in special cases, 
a prediction that may be applied to new problems. An obvious choice is the 
cardinality $m$ of the set of possible requests. (Clearly, this question is interesting 
only for $m\ll n$, otherwise the known general lower bounds kick in.) Examining 
request sequences drawn from a given set is proposed explicitly in~\cite{BI97}. 
They give a deterministic algorithm that in uniform metric spaces achieves a 
competitive ratio of $O(\log n)$ times the (unknown) best possible bound.
\newline

Motivated by the $k$-server example, the restriction of the requests to have
costs in $\{0,\infty\}$ is a natural and attractive problem of chasing subsets
of points in a finite metric space. Indeed, this problem setting was proposed
in~\cite{CL93}, dubbed metrical service systems. We note that this is part
of a broader theme of online chasing and searching, and other seminal
papers in this theme include~\cite{PY89,BCR93,FL93}. The focus of the metrical
service systems literature (e.g.~\cite{CL93,FFKRRV91,Ram93,Bur96,Feu98,CV15}) 
is the {\em width} of the requests, which is the maximum size of a requested set.
In view of the above discussion, it is equally natural to ask about chasing 
sets drawn from a small-sized pool of possible requests and deriving bounds
that are independent of the width.

\paragraph{Our results.}
Our main contribution is an explicit analysis of the competitive ratio of metrical
task systems and metrical service systems with $m$ possible requests, $m\ll n$. 
Perhaps surprisingly, we show that unlike the general case, the answer here 
depends on the metric space. In particular, for metrical task systems we show 
that in uniform metric spaces, restricting thus the request sequence improves 
dramatically the deterministic competitive ratio, which we characterize tightly, up 
to constant factors. In contrast, there exist other metric spaces, namely two-level 
HSTs (uniform metrics are one-level HSTs), that exhibit no deterministic asymptotic 
improvement, even if $m=2$. Perhaps even more surprisingly, in uniform metric 
spaces the randomized competitive ratio does not improve asymptotically even if 
$m=2$, and in this case the deterministic competitive ratio is asymptotically the 
same as the randomized competitive ratio.
\newline

We also analyze unrestricted width metrical service systems with $m$ possible
requests in uniform metric spaces versus two-level HSTs. This case, too,
exhibits complex and fascinating behavior, and furthermore it is not identical to
the general case. In sufficiently large uniform metric spaces the restriction to $m$
requests improves both the deterministic and the randomized competitive ratio.
However the two ratios are asymptotically the same, linear in $m$ and independent
of the size of the metric. The two-level HSTs case also deviates from the general
case. As with uniform metric spaces, both the deterministic and the randomized
competitive ratios improve and are independent of the size of the metric space.
However, both the deterministic and the randomized competitive ratios are now
exponential in $m$ (the base of the exponent may be different in the two cases).
We discuss the surprising aspects of these results and their features that persist
beyond two levels in the following sections.
\newline

We note that our results, most strikingly the bounds for $m=2$ metrical task
systems in uniform metric spaces, refute a folklore conjecture that for most
problems that can be cast as metrical task systems, the randomized competitive
ratio is polylogarithmic in the deterministic competitive ratio.
\newline

Finally, we relate some of our results to the question of quantifying the amount
of information an adversary needs in order to force an online algorithm to perform
poorly.

\section{Definitions and Results}

A {\em metrical task system} instance consists of a finite metric space ${\cal M} = (X,d)$,
an initial state $s_0\in X$, and a sequence of requests $\rho_1,\rho_2,\dots,\rho_L$,
where for every $t\in\{1,2,\dots,L\}$, the request $\rho_t$ is a cost function that
maps $X$ to $\RR_+ \cup \{+\infty\}$. We will denote $n \triangleq |X|$. A solution consists of a 
choice of states $s_1,s_2,\dots,s_L\in X$, which incurs a total cost of
$$
\sum_{t=1}^L \left(d(s_{t-1},s_t) + \rho_t(s_t)\right).
$$
A deterministic online algorithm chooses, for every $t=1,2,\dots,L$, the next state 
$s_t$ based only on ${\cal M}$ and $\rho_1,\rho_2,\dots,\rho_t$, without knowing 
the suffix of the requests sequence or even its length $L$. A randomized online
algorithm is a probability distribution over deterministic online algorithms, and its
cost is defined to be the expectation of the cost of the deterministic algorithms in
the support of the distribution.

We use the following notation. For a sequence of requests $\rho$, let $z^*(\rho)$
denote the optimal cost of serving $\rho$. Define $c^{\detr}_{\cal M}$ to be the 
infimum over $c$ such that there exists a constant $a$ and a deterministic online 
algorithm $A$ such that for every sequence of requests $\rho$, the cost of $A$ of
serving $\rho$ is at most $c\cdot z^*(\rho) + a$. Define similarly $c^{\rand}_{\cal M}$
for randomized algorithms. The following theorem is well-known:
\begin{theorem}[Borodin et al.~\cite{BLS92}]\label{thm: det}
For every finite metric space ${\cal M}$,
$$c^{\detr}_{\cal M} = 2n-1.$$
\end{theorem}
It is also conjectured that $c^{\rand}_{\cal M} = \Theta(\log n)$. The conjecture has
been established in special cases. For instance tight bounds are known in uniform
metric spaces. The currently best known bounds (following substantial previous
work~\cite{KRR91,BKRS00,BI97,BB97,BBBT97,IS98,Sei99,FM00,BBN10,ABBS10})
that hold for every metric space are (see below for the definition of an HST):
\begin{theorem}[Bartal et al.~\cite{BBM01,BLMN03}]\label{thm: rand lower}
For every finite metric space ${\cal M}$,
$$c^{\rand}_{\cal M} = \Omega\left(\frac{\log n}{\log\log n}\right).$$
\end{theorem}
\begin{theorem}[Bubeck et al.~\cite{BCLL18}]\label{thm: rand upper}
For every HST ${\cal M}$,
$$c^{\rand}_{\cal M} = O(\log n),$$
and therefore, for every finite metric space ${\cal M}$,
$$c^{\rand}_{\cal M} = O(\log^2 n).$$
\end{theorem}

We are interested in analyzing metrical task systems with constrained requests.
The primary constraint that we investigate in this paper
is the following: there exists a finite set ${\cal R}\subsetneq \left(\RR_+\cup\{+\infty\}\right)^X$
such that all the requests lie in ${\cal R}$. We will denote $m \triangleq |{\cal R}|$.
Let $c^{\detr}_{\cal M, R}$ ($c^{\rand}_{\cal M, R}$, respectively) denote the
deterministic (randomized, respectively) competitive ratio when requests are
restricted to the set ${\cal R}$. Also let $c^{\detr}_{{\cal M}, m}$
($c^{\rand}_{{\cal M}, m}$, respectively) denote the supremum of
$c^{\detr}_{\cal M, R}$ ($c^{\rand}_{\cal M, R}$, respectively) over $|{\cal R}| = m$.
We find that unlike the general case, the competitive ratio that can be guaranteed
for constrained metrical task systems depends crucially on the underlying metric
space ${\cal M}$. In particular, we study a class of ultrametrics called hierarchically
separated trees (HSTs), invented in~\cite{Bar96}. HSTs play a critical role in both the 
best known upper bounds and the best known lower bounds for metrical task systems,
as well as other problems involving metric spaces. For metrical task systems, the
best known upper bounds rely on an asymptotically optimal approximation of any
metric space by a convex combination of HSTs, discovered in~\cite{FRT04}. The
best known lower bounds rely on a lower bound on HSTs and the fact that any
metric space contains a large approximate HST subspace~\cite{BBM01,BLMN03}.
More concretely, an $L$-level HST is a metric space defined on the leaves of a 
leveled rooted node-weighed tree with $L+1$ levels. The leaves are at level $0$, 
the root is the unique node at level $L$, each node other than the root has a parent 
at one level above
its own, all the nodes at level $i$ have the same weight, and this weight increases 
rapidly with the level number. The distance between two leaves is the weight of their 
least common ancestor. (We note that HSTs are often defined alternatively without 
the uniformity constraint on levels.) This defines an ultrametric (which is in particular 
a metric) on the leaves.

We exhibit on the one hand that for uniform metric spaces (which are in the above
terms one-level HSTs), restricting $m$ helps immensely deterministic algorithms,
but not randomized algorithms:
\begin{theorem}\label{thm: uniform}
If ${\cal M}$ is a uniform metric space, then for every $m\le n$,
$$c^{\detr}_{{\cal M},m} = \Theta(m\log(en/m)),$$
and
$$c^{\rand}_{{\cal M},2} = \Theta(c^{\rand}_{\cal M}) = \Theta(\log n).$$
\end{theorem}

On the other hand, we study what is in some sense just a slightly more sophisticated 
family of metric spaces. We consider the following family of two-level HSTs, which we
call {\em paired-uniform}. Let $n$ be an even number of points. Partition the points 
into pairs $\{i,i'\}$. (So, according to this notation $i'' = i$.) Let $C\gg 1$ be a large 
constant to be determined later. Set $d(i,i') = 1$ for all pairs, and set $d(i,j) = C$ for 
every two points $i\ne j\ne i'$. It turns out that here restricting $m$ helps neither the 
asymptotic performance of deterministic algorithms, nor that of randomized algorithms:
\begin{theorem}\label{thm: paired-uniform}
If ${\cal M}$ is a paired-uniform metric space, then
$$c^{\detr}_{{\cal M},2} = \Theta(n),$$
whereas 
$$c^{\rand}_{{\cal M},2} = \Theta(\log n).$$
\end{theorem}
Notice that by prior results, the upper bounds hold even for arbitrary requests
in arbitrary HSTs.

The lower bound constructions in the proofs of Theorems~\ref{thm: uniform}
and~\ref{thm: paired-uniform} use sets of requests that assign many widely
varying cost values. Thus, it makes sense to consider requests using costs
from a small set of scales, and in particular to consider metrical service systems
that use just two scales, $0$ or $\infty$ (a.k.a. set chasing). Let
$\hat{c}^{\detr}_{{\cal M},m}$, ($\hat{c}^{\rand}_{{\cal M},m}$, respectively) 
denote the best deterministic (randomized, respectively) competitive ratio that 
can be achieved for chasing sets drawn from a collection ${\cal R}$ of $m$ 
subsets of points in ${\cal M}$. In this case we show that for uniform spaces, 
restricting $m$ helps {\em both} deterministic and randomized algorithms. 
Moreover, if $m = O(\log n)$, the deterministic and the randomized competitive
ratios are asymptotically identical. More precisely we have:
\begin{theorem}\label{thm: 0-infty uniform}
If ${\cal M}$ is a uniform metric space, then for every $m\in\NN$,
$$\hat{c}^{\detr}_{{\cal M},m} = \min\{m,n\},$$
and,
$$\hat{c}^{\rand}_{{\cal M},m} = \Theta\left(\min\{m,\log n\}\right).$$
\end{theorem}

Unlike the more general case of Theorem~\ref{thm: paired-uniform}, in the case
of metrical service systems on two-level HSTs, restricting the number of possible 
requests helps, but not as much as in uniform metric spaces. In order to state the
results, we introduce a bit more notation. Let $\hat{c}^{\detr}_{L,m}$ 
($\hat{c}^{\rand}_{L,m}$, respectively) 
denote the supremum of $\hat{c}^{\detr}_{{\cal M},m}$ 
($\hat{c}^{\rand}_{{\cal M},m}$, respectively) over all $L$-level HSTs ${\cal M}$.
We prove:
\begin{theorem}\label{thm: 0-infty two-level}
For every $m\in\NN$,
$$2^{\lfloor m/2\rfloor-1} \le  \hat{c}^{\rand}_{2,m} \le \hat{c}^{\detr}_{2,m} \le m2^m.$$
\end{theorem}
We note that we actually prove a somewhat tighter lower bound of ${m \choose \lfloor m/2\rfloor}$
on $\hat{c}^{\detr}_{2,m}$, see Lemma~\ref{lm: 0-infty two-level detr lower}.
We also note that it is not surprising that for sufficiently large HSTs, the competitive
ratio (deterministic or randomized) depends solely on $m$ and the number of levels.
This is because one can trim isomorphic nodes from the tree.
(Note however that for deterministic algorithms the trivial bound obtained this way is 
exponentially worse than the ones proposed above.) 
The surprising aspects
of these theorems are (a) that the deterministic competitive ratio nearly matches the 
randomized competitive ratio, and 
(b) that for two-level HSTs, a collection of $O(\log\log n)$ set chasing requests is 
sufficient to generate a randomized lower bound of $\Omega(\log n)$ (while the 
classical lower bound requires $\Omega(n)$ such requests). In fact, this latter aspect 
extends to and is amplified in HSTs with any number of levels. So, using just $6$ set 
chasing requests, we can get a lower bound of $\Omega(\log n)$ for infinitely many 
$n$, on HSTs with more and more levels. The argument is outlined in 
Section~\ref{sec: discussion}.

Our main results are summarized in Table~1.

\begin{center}
\begin{table}
\begin{center}
\begin{tabular}{ccc}
one-level HSTs (uniform metrics) & & two-level HSTs \\ \\
\begin{tabular}{|c||c|c|}
\hline
 & deterministic & randomized \\
 \hline
 \hline
 MTS & $\Theta(m \log n)$ & $\Theta(\log n)$ \\
 \hline
 MSS & $\Theta(m)$ & $\Theta(m)$ \\
 \hline
 \end{tabular}
 &
 &
 \begin{tabular}{|c||c|c|}
 \hline
 & deterministic & randomized \\
 \hline
 \hline
 MTS & $\Theta(n)$ & $\Theta(\log n)$ \\
 \hline
 MSS & $2^{\Theta(m)}$ & $2^{\Theta(m)}$ \\
 \hline
 \end{tabular}
 \end{tabular}
 \end{center}
 \label{tab: table1}
 \caption{Summary of our results, for any $n\gg m$.}
\end{table}
\end{center}

\section{Proofs}

\begin{lemma}\label{lm: good s}
Let ${\cal M} = (X,d)$ be an $n$-point uniform metric space.
Consider any set of requests ${\cal R}$ of cardinality $|{\cal R}| = m$ and
any set $S\subseteq X$ of cardinality $|S| > m\log(n/m)$. Then, there exists a 
state $s\in S$ such that for every $r\in {\cal R}$,
\begin{equation}\label{eq: eliminate}
\left|\{s'\in S:\ r_{s'}\ge r_s\}\right|\ge \frac{1}{m}\cdot |S|.
\end{equation}
\end{lemma}

\begin{proof}
The proof is a simple application of the probabilistic method. Choose
$s\in S$ uniformly at random. We have that for every $r\in {\cal R}$,
$$\Pr\left[|\{s'\in S:\ r_{s'}\ge r_s\}| < \frac{1}{m}\cdot |S|\right] < \frac{1}{m}.$$
Therefore, applying the union bound,
$$\Pr\left[\exists r\in {\cal R}:\ |\{s'\in S:\ r_{s'}\ge r_s\}| < \frac{1}{m}\cdot |S|\right] < 1,$$
so there is a choice of $s\in S$ that satisfies Equation~\eqref{eq: eliminate}.
\end{proof}

\begin{lemma}\label{lm: det upper}
Let ${\cal M}$ be an $n$-point uniform metric space. Then,
$c^{\detr}_{{\cal M},m} = O(m\log(en/m))$.
\end{lemma}

\begin{proof}
Consider the following algorithm. The algorithm works in phases, and
these are further partitioned into rounds. In a phase, we will denote by 
$S_r$ the set $S$ at the end
of round $r$. We define the cost that a state $s$ acrues during some
interval $I$ as $\Sigma(I; s) = \sum_{t\in I} \rho_t(s)$.
In the beginning of a phase, set $S_0 = X$ and move to a 
state $s$ that satisfies Equation~\eqref{eq: eliminate}. As soon as at least 
$|S_{r-1}| / 2m$ states in $S_{r-1}$ acrue a cost of at least $1$ in the
current phase (including the new request), round $r$ ends. When this 
happens, remove from $S_{r-1}$ the $\left\lceil |S_{r-1}| / 2m\right\rceil$ states that acrued
the highest cost. Set $S_r$ to be the remaining states in $S_{r-1}$. 
If $|S_r| \ge m\log(en/m)$, repeat the above 
process with a new $s\in S_r$. Otherwise, this is the last round. In this last
round, use an $|S_r|$-competitive algorithm, restricted to $S_r$. Execute
this algorithm until the optimal cost of serving this round reaches at least $1$.
Notice that this implies that by the end of the phase, all states in $X$ acrued
a cost of at least $1$ during the phase. Therefore, the adversary's cost for
the entire phase must be at least $1$, because it either stayed at one state
and paid a service cost of $1$, or moved at least once and paid a movement
cost of $1$. 

In order to analyze this algorithm, fix a phase and denote by $a_r$ the total 
cost that the algorithm acrues in round $r$ of this phase. Also denote by $o_r$
the average over $s\in S_r$ of the minimum between $1$ and the cost acrued 
by $s$ in this phase by the end of round $r$. (Set $o_0 = 0$.) Notice that 
$|S_r| \le \left(1 - \frac{1}{2m}\right)\cdot |S_{r-1}|$. This implies that the number
of rounds (including the special last round) is at most $2m\log(n/m)+1$.
Also notice that throughout
round $r$, at each request at least $|S_{r-1}|/m$ states in $S_{r-1}$ pay at least 
the cost that the algorithm pays, but less than $|S_{r-1}|/2m$ states in $S_{r-1}$
acrued a total cost of at least $1$. Thus, the average over $S_{r-1}$ increases
by at least $a_r / 2m$. Each state in $S_{r-1}\setminus S_r$ contributes
$\frac{1}{|S_{r-1}|}$ to this average, and there are $|S_{r-1}|/2m$ such states,
so removing them reduces the average by at most $\frac{1}{2m}$. Thus,
$$
o_r \ge o_{r-1} + \frac{a_r}{2m} - \frac{1}{2m}.
$$
Let $\bar{r}$ be the penultimate round. Notice that for all $r$, $o_r\le 1$. 
Then 
$$
2m \ge 2m\cdot o_{\bar{r}}\ge \sum_{r=1}^{\bar{r}} a_r - \bar{r} \ge
\sum_{r=1}^{\bar{r}} a_r - 2m\log(n/m).
$$
Thus, $\sum_{r=1}^{\bar{r}} a_r = O(m\log(en/m))$. The last phase
also costs the algorithm $O(m\log(en/m))$. This is because 
$|S_{\bar{r}}| < m\log(en/m)$, we use an $|S_{\bar{r}}|$-competitive
algorithm, and an adversary with the same initial state in this round
as the algorithm pays at most $2$---one for moving to the best initial
state, and one for serving the request sequence in the round.
\end{proof}

\begin{lemma}\label{lm: lower bound m=2}
Let ${\cal M}$ be an $n$-point uniform metric space. Then,
$c^{\rand}_{{\cal M},2} = \Omega(\log n)$.
\end{lemma}

\begin{proof}
Assume without loss of generality that $n=2^q$ is a power of $2$
(otherwise we can restrict the instance to the power of $2$ closest
to $n$; the redundant states can be blocked by assigning to them
a cost of $\infty$ in both elements of ${\cal R}$).
Let $\{0,1,\dots,n-1\}$ be an arbitrary enumeration of the points
of ${\cal M}$. Let $C = C(n)$ be a large constant.
Set ${\cal R} = \{r^0,r^1\}$, where $r^0_i = C^{i-n}$, and let
$r^1_i = C^{-i-1}$.
We use the minimax principle and generate a random request
sequence that ``beats'' every deterministic algorithm. The request
sequence is generated in phases. In each phase we choose uniformly
at random one state $h\in\{0,1,\dots,n-1\}$ where the adversary 
``hides.'' Let $h_0 h_1 \cdots h_{q-1}$ be the binary representation
of $h$. This defines a natural sequence of nested intervals in the
dyadic partition of the set of states. The interval $J_i$ consists of
all the states whose binary representation has the prefix $h_0\cdots h_i$.
For notational convenience, let $J_{-1} = \{0,1,\dots,n-1\}$.
The phase consists of $q$ rounds numbered $0,1,\dots,q-1$. 
Then, in round $i$, we request repeatedly $r^{h_i}$ until
the state adjacent to $J_i$ on the right, if $h_i = 0$, or on the left,
if $h_i = 1$, acrues a cost of at least $1$ in this round. Notice that 
the algorithm knows the entire round once its first request is given.

The analysis is completed as follows. Define 
$$\bar{J}_i = \left\{\begin{array}{ll}
                             \{j: j\le\max s\in J_i\} & h_i = 0, \\
                             \{j: j\ge\min s\in J_i\} & h_i = 1.
                             \end{array}\right.
$$
When round $i$ starts, with probability at least $\frac 1 2$, the algorithm 
occupies a state $s\not\in \bar{J}_i$. In this case the algorithm pays at
least $1$ in this round. This is clear if it moves. Otherwise, every state
not in $\bar{J}_i$ pays at least as much as the state in this set that is
adjacent to $J_i$. Thus, the expected cost of the algorithm for the entire 
phase is at least $q/2 = \log(n-1)$. On the other hand, the adversary
pays at most $1$ to move to $h$. Notice that the state
adjacent to $J_i$ that determines the stopping condition pays less than
$2$ in round $i$ (in fact it pays close to $1$). Thus, all the states in $J_i$ 
pay at most $\frac{2}{C}$. If the adversary stays at $h$ for the 
duration of the phase, it pays a service cost of at most $\frac{2q}{C}$,
which can be made arbitrarily close to $0$ by choosing $C\gg q$.
\end{proof}

\begin{lemma}\label{lm: det lower bound}
Let ${\cal M}$ be an $n$-point uniform metric space. Then,
$c^{\detr}_{{\cal M},m} = \Omega(m \log(en/m))$.
\end{lemma}

\begin{proofsketch}
Without loss of generality, we may assume that $m$ is an even number which
is at most $n/2$, that $2n$ is divisible by $m$ and that $2n/m$ is a power of $2$.
(Otherwise, eliminate some states by having their cost always $\infty$ and use
fewer requests. A constant fraction of both states and requests remains.)
So, partition the 
states into $m/2$ disjoint subsets of size $2n/m$ each. For each subset,
define two new requests as per the proof of Lemma~\ref{lm: lower bound m=2},
and outside the subset have both requests cost $0$. In each subset, use
the adversary's strategy from the proof of Lemma~\ref{lm: lower bound m=2}.
The requests according to this strategy will be used in the steps when the
algorithm occupies a state in the subset. Clearly, while the algorithm hasn't
paid at least $\Omega(\log(n/m))$ in all subsets, there's at least one state
that hasn't paid $1$. Thus, the adversary can choose the last such surviving
state and cause the algorithm to pay $\Omega(m\log(n/m))$ before no state
survives. When no state survives, a phase that cost the adversary less than
$2$ and the algorithm $\Omega(m\log(n/m))$ ends, and we can repeat the
process in a new phase.
\end{proofsketch}

\begin{proofof}{Theorem~\ref{thm: uniform}}
The theorem is a corollary of Lemmas~\ref{lm: det upper},~\ref{lm: lower bound m=2},
and~\ref{lm: det lower bound}.
\end{proofof}

\begin{proofof}{Theorem~\ref{thm: paired-uniform}}
The upper bound on $c^{\rand}_{{\cal M}, 2}$ follows from the general upper bound
on trees of Bubeck et al.~\cite{BCLL18}, and the lower bound follows from
Theorem~\ref{thm: uniform}. The lower bound on $c^{\detr}_{{\cal M},2}$
goes as follows. Denote the pairs of leaves by $i, i'$ for $i=0,1,\dots,n/2-1$. Let
$h_0,h_1,\dots,h_{n/2-1}$ be the sequence defined by $h_i = (Cn)^{i-n}$. We use
two requests: $r$ is defined by $r_i = h_i$ and $r_{i'} = 0$, for all $i$, and $r'$
is defined by $r'_i = 0$ and $r'_{i'} = h_{n-i-1}$, for all $i$. The request sequence
is simple: if the algorithm occupies a state $i$, use $r$, otherwise the algorithm
occupies a state $i'$, so use $r'$.

For analyzing the competitive ratio, partition the request sequence into rounds.
A round ends whenever the algorithm either moves from one pair to another
pair or pays $C$ while staying in one pair $\{i,i'\}$. Notice that either way, the
algorithm's cost is at least $C$ per round. Further partition the rounds into phases.
Each phase contains exactly $n/2-1$ rounds. Thus, the cost of the algorithm per
phase is at least $C(n/2-1)$. We show that the adversary can pay less than
$C + n/2$ per round. Choosing $C = \Omega(n)$ completes the proof. 
To show that the adversary's cost per phase is less than $C + n/2$, 
notice that in a
phase, the algorithm cannot visit all the pairs. Let $\{i,i'\}$ be a pair that is not
visited by the algorithm during the phase. The adversary will ``hide'' in such a pair.
Moving to this pair in the beginning of the round costs at most $C$. For each 
round in the phase, the adversary stays either at $i$ or at $i'$, so moving between
$i$ and $i'$ costs at most $n/2-1$. 

To complete the analysis, we have to explain
how the adversary chooses between $i$ and $i'$, and we have to analyze the
cost of staying at the chosen state. Notice that in a round, by definition the 
algorithm stays in one pair $\{j,j'\}$, where $i\ne j$. If $i < j$, the adversary
stays at $i$, otherwise the adversary stays at $i'$. Consider the first case
(the other case is analogous). By the definition of a round, $j$ acrues a cost
of at most $C$ during the round, because it is hit by a positive cost only when
the algorithm occupies $j$, and the algorithm pays at least this cost (it might
decide to move to $j'$ and pay $1$). Therefore, $i$ acrues a cost of at most
$C\cdot (Cn)^{i-j}\le\frac{1}{n}$. This is the service cost that the adversary
pays during the round. Thus, since there are $n/2-1$ rounds in a phase, the
total service cost of the adversary during a phase is less than $1$.
\end{proofof}

\begin{lemma}\label{lm: 0-infty uniform detr}
Let ${\cal M}$ be an $n$-point uniform metric space. Then,
$\hat{c}^{\detr}_{{\cal M},m} = \min\{m,n\}$.
\end{lemma}

\begin{proof}
The lower bound follows from restricting the number of states 
that can get a request with a $0$ value to $k+1 = \min\{m,n\}$ and then applying 
the deterministic $k$-server lower bound to this subspace. Notice that
this lower bound uses $k+1$ different requests.

The upper bound is achieved with a variant of the marking algorithm, as follows. 
Each request $r$ can be associated
with a subset $S_r = \{i\in X:\ r_i = 0\}$. If there exists a state $i$ that is in the
intersection of all $m$ subsets $S_r$, the algorithm can pay at most $1$ to 
move to this state, and then pay $0$ for the rest of the sequence. Otherwise,
the algorithm partitions the sequence into phases. In each phase, if the current
location of the algorithm is hit by an $\infty$ value, it moves to a state in the
intersection of all the sets seen so far in the phase. The phase ends, and a
new phase begins, when this intersection is empty. Notice that to move, the
algorithm needs to get a request not seen so far in the phase, so the algorithm
moves at most $m$ times per phase. At the end of the phase, all states were
hit by an $\infty$ value at least once, so any algorithm would have to move at
least once per phase.
\end{proof}

\begin{lemma}\label{lm: 0-infty uniform rand upper}
Let ${\cal M}$ be an $n$-point uniform metric space. Then,
$\hat{c}^{\rand}_{{\cal M},m} = O\left(\min\{m,\log n\}\right)$.
\end{lemma}

\begin{proof}
Theorem~\ref{thm: rand upper} implies in particular an upper bound of $O(\log n)$.
Lemma~\ref{lm: 0-infty uniform detr} gives an upper bound of $m$ (because the
randomized competitive ratio is upper bounded by the deterministic competitive
ratio).
\end{proof}

\begin{lemma}\label{lm: 0-infty uniform rand lower}
Let ${\cal M}$ be an $n$-point uniform metric space. Then,
$\hat{c}^{\rand}_{{\cal M},m} = \Omega\left(\min\{m,\log n\}\right)$.
\end{lemma}

\begin{proof}
If $m\ge n$, then we get a lower bound of ${\cal H}_{n-1}$ from the paging 
problem~\cite{FKLMSY91}, which requires $n$ possible requests. Otherwise, let
$k = \left\lfloor\min\left\{m/2,\log_2 n\right\}\right\rfloor$. 
We restrict our attention to a set of size $2^k$ states. 
(If there are more, all the others have cost $\infty$ in all the requests.) We label 
the states in this set with the nodes of the binary cube $\{0,1\}^k$. We use $2k$
possible requests that are constructed as follows. There are two requests $r_{i,0}$ 
and $r_{i,1}$ corresponding to each coordinate $i=1,2,\dots,k$ of the binary cube. 
The request $r_{i,b}$ has a cost of $\infty$ at points with the $i$-th coordinate of 
their label being $b$, and $0$ at states with the complement $i$-th coordinate. We
use the minimax principle and construct a probabilistic request sequence that beats
every deterministic algorithm. The request sequence consists of phases. In each 
phase the adversary draws independently a point $b_1b_2\cdots b_k$ in the binary 
cube uniformly at random, and then requests the sequence $r_{i,b_i}$, $i=1,2,\dots,k$, 
in that order. Against any deterministic algorithm, the expected cost of each request 
in the phase is $\frac 1 2$, so the total expected cost of $\frac k 2$. The adversary 
needs to move at most once per phase, to the point labeled with the bitwise complement 
of $b_1b_2\cdots b_k$.
\end{proof}

\begin{proofof}{Theorem~\ref{thm: 0-infty uniform}}
The theorem is an immediate corollary of
Lemmas~\ref{lm: 0-infty uniform detr},~\ref{lm: 0-infty uniform rand upper},
and~\ref{lm: 0-infty uniform rand lower}.
\end{proofof}

\begin{lemma}\label{lm: 0-infty two-level detr upper}
For every two-level HST ${\cal M}$, $\hat{c}^{\detr}_{{\cal M},m} \le m2^m$.
\end{lemma}

\begin{proof}
Let $C > 1$ denote the aspect ratio of ${\cal M}$. Let $T$ be any level-$1$ subtree
of ${\cal M}$ (so in particular $T$ indicates a uniform subspace of ${\cal M}$). We
say that $T$ is hit by a subset $S$ of requests iff for every state $s\in T$ there is
a request $r\in S$ such that $r_s = \infty$. The algorithm works in periods. In the
beginning of a period all level-$1$ subtrees are unmarked, and a period ends when
they are all marked. Each period is divided into epochs. In an epoch, the algorithm
chooses an unmarked subtree $T$ and stays in it for the entire epoch. While 
in  a subtree $T$, the algorithm runs the uniform metric procedure from the proof 
of Lemma~\ref{lm: 0-infty uniform detr}, ignoring the states outside $T$. Notice
that this divides the epoch into phases, where during a phase the algorithm
encounters one subset $S$ of requests that hits $T$, for which the algorithm
pays $|S|\le m$ and any algorithm staying in $T$ pays at least $1$. The epoch 
ends as soon as the algorithm executes $C$ phases in $T$. When the epoch ends,
the algorithm marks $T$ and any other subtree that was hit at least $C$ times
so far in the period. The cost of the algorithm in a period is upper bounded as
follows. In each phase the algorithm encounters a hitting set $S$. If the same
set is encountered $C$ times, all the subtrees that it hits get marked. Therefore,
since the empty set is not a hitting set, the total number of phases is at most
$C(2^m-1)$. The number of epochs in a period is therefore $2^m-1$, and thus
the total cost of transition between subtrees is $C(2^m-2)$ (because no transition
is required at the end of the last epoch in the period). Since the algorithm
encounters each non-empty set $S$ at most $C$ times, the total ``internal"
cost of the phases is at most $C\cdot\sum_S |S| = Cm2^{m-1}$. Therefore,
the total cost of the algorithm per period is $C(m2^{m-1}+2^m-2) < Cm2^m$,
for all $m\in\NN$. Any algorithm must pay at least $C$ per period. If it stays
in one subtree, then every subtree is hit at least $C$ times. If it moves between
subtrees, this incurs a cost of at least $C$.
\end{proof}

\begin{lemma}\label{lm: 0-infty two-level detr lower}
For all $m\in\NN$, $\hat{c}^{\detr}_{2,m} \ge {m \choose \lfloor m/2\rfloor}$.
\end{lemma}

\begin{proof}
We construct the following two-level HST with aspect ratio $C$ and request set ${\cal R}$. 
There are ${m \choose \lfloor m/2\rfloor}$ subtrees, each labeled by a different subset
of the requests of cardinality $\lfloor m/2\rfloor$. Each subtree has $\lfloor m/2\rfloor$
states labeled by the elements of the subtree's label. A request has cost $\infty$ at
all the states that are labeled by it, and $0$ otherwise. The adversary's strategy is
simple: if the algorithm is currently at a subtree labeled $S$, apend to the sequence
all the requests in $S$ (the order doesn't matter), then repeat after updating the
location that the algorithms reaches after the apended requests. In order to analyze
the competitive ratio, partition the request sequence into periods, where in each
period the algorithm leaves all the subtrees. Fix a period, and let $k$ denote the
number of times the algorithm changes a subtree in this period, including the
departure that ends the period. Clearly, $k\ge {m \choose \lfloor m/2\rfloor}$.
The algorithm pays at least $Ck$ in this period. The adversary can hide at the
last tree that is visited. At each step, no matter which subtree the algorithm
occupies, there is at least one state in the last tree that won't be hit by the apended
requests. Each time the algorithm switches trees, the adversary must move to
a new state in this last subtree. At the end of the period, the adversary may have
to move to a new subtree. So, the total cost of the adversary per period is at most
$k + C$. If we choose $C={m \choose \lfloor m/2\rfloor}$, this gives a lower bound 
of $\frac 1 2\cdot {m \choose \lfloor m/2\rfloor}$. However, we can get a lower
bound arbitrarily close to ${m \choose \lfloor m/2\rfloor}$ by choosing a larger $C$.
\end{proof}

\begin{lemma}\label{lm: 0-infty two-level rand lower}
For all $m\in\NN$, $\hat{c}^{rand}_{2,m} \ge 2^{\lfloor m/2\rfloor-1}$.
\end{lemma}

\begin{proof}
The argument is basically a ``lifting" of the construction in the proof of
Lemma~\ref{lm: 0-infty uniform rand lower}. As in that proof, we associate the
requests with the coordinates of the binary cube $\{0,1\}^{m/2}$ (we assume 
without loss of generality that $m$ is even, otherwise, discard one request).
Each state is labeled by a node of the cube, and request $r_{i,b}$ has cost
$\infty$ for all states with a label that has value $b$ in its $i$-th coordinate,
and cost $0$ otherwise. The requests are paired into pairs $r_{i,0},r_{i,1}$.
Notice that each request hits half of the possible state labels, and the sets
of labels that two paired requests hit are complements. Now, consider a
sequence of requests generated by choosing one request from each pair.
We can label this sequence with a point $b$ in the binary cube, where
the sequence requests $r_{i,b_i}$ for $i=1,2,\dots,m/2$. Thus, there are
$2^{m/2}$ possible such sequences. If such a sequence is requested, it
hits all the labels, except for one---the bitwise complement of $b$. Thus,
we can pair such sequences into pairs that have complement labels, and
the two sequences in a pair miss complement labels of states. There are
$2^{m/2-1}$ pairs of sequences. We can create a meta-sequence of
requests by choosing one sequence in each pair and concatenating all
the chosen sequences. This gives $2^{2^{m/2-1}}$ possible meta-sequences.
This structure is used to generate the adversary's strategy.

But, before we state the adversary's strategy, we construct the metric space
and set the labels of the states (which imply the structure of the individual
requests). The two-level HST that we use has $2^{2^{m/2-1}}$ subtrees.
A subtree is labeled by a subset of the binary cube $\{0,1\}^{m/2}$ of
cardinality $2^{m/2-1}$ (half the cube) that contains exactly one node of
each antipodal (i.e., complementary) pair of nodes. There are $2^{2^{m/2-1}}$
possible choices, hence the number of subtrees. Such a subtree has
$2^{m/2-1}$ states, each labeled with a distinct node in the label of the
subtree. Aside from the aspect ratio $C$, to be defined later, this specifies
completely ${\cal M}$ and ${\cal R}$.

We are now ready to define the adversary's strategy. As usual, we rely on
the minimax principle and define a randomized strategy that beats any
deterministic algorithm. The adversary repeats the following process.
Choose one meta-sequence uniformly at random, then loop through
its list of sequences and repeat each sequence $C$ times. Notice that
regardless of the state of the algorithm at the beginning of any sequence
of the meta-sequence, it pays at least $C$ for this sequence with probability
at least $\frac 1 2$. This is because each sequence hits half the subtrees,
and the two possible choices for the sequence hit complement sets of
subtrees: a sequence labeled $b$ hits all subtrees that contain $b$,
a sequence labeled $\bar{b}$ hits all subtrees that contain $\bar{b}$,
and every subtree contains either $b$ or $\bar{b}$, regardless of $b$.
If the sequence hits the subtree where the algorithm is located, it either
stays there and pays at least $1$ for each time the sequence is requested
(because it hits all the states in that subtree), or it moves to a different tree
and pays $C$ for the transition. On the other hand, after requesting the
entire meta-sequence, there exists one subtree that was not hit even
once. This is the subtree that is labeled with the set of complements of
all the labels of the sequences in the meta-sequence. The adversary can
hide in that subtree. It pays $C$ to move there at the start of the current
meta-sequence, and at most $1$ to move to a safe state at the start of
each sequence in the meta-sequence. So, the expected cost of the
algorithm per meta-sequence is $C2^{m/2-1}$, and the total cost of the
adversary per meta-sequence is $C + 2^{m/2-1}$. As $C$ grows, the
ratio approaches $2^{m/2-1}$.
\end{proof}

\begin{proofof}{Theorem~\ref{thm: 0-infty two-level}}
The theorem is a corollary of Lemmas~\ref{lm: 0-infty two-level detr upper}
and~\ref{lm: 0-infty two-level rand lower}.
\end{proofof}

\section{Discussion}\label{sec: discussion}

\paragraph{Open problems.}
Our work initiates the study of metrical task systems (MTS) and metrical service systems (MSS)
parametrized by the number of distinct requests. Roughly speaking (and somewhat 
surprisingly) we find that this restriction has little effect in general for MTS, 
in the sense that beyond uniform metric spaces, the achievable competitive 
ratio with $m=2$ is already asymptotically as bad as for $m=\infty$ 
(Theorem~\ref{thm: paired-uniform}). In fact, as far as randomized algorithms are 
concerned this is already true for uniform metric spaces (Theorem~\ref{thm: uniform}). 
On the other hand, the situation for MSS is very different. A number 
of questions remain open regarding MSS. A particularly interesting 
qualitative open problem would be to characterize the class of infinite size metric spaces 
for which there exist online MSS algorithms with finite competitiveness for fixed $m$,
say even for $m=3$. For example, is it possible to obtain a finite competitive ratio for 
chasing $3$ arbitrary sets on the real line? Another intriguing quantitative open problem 
is to determine if the deterministic competitive ratio on a weighted star metric is linear
or exponential in $m$. (A simple argument bounds the deterministic competitive ratio
by $O(2^m)$ and the randomized competitive ratio by $O(m)$. A lower bound of $\Omega(m)$
in both cases clearly follows from Theorem~\ref{thm: 0-infty uniform}.)

\paragraph{Deeper trees.}
Consider MSS on an HST. Label each leaf by a binary vector of dimension $m$ indicating
which requests hit it. Clearly, if there are two identically labeled leaves that share a
parent, we can eliminate one from consideration. Similarly, label each internal node by the
set of labels of its children. If there are two identically labeled internal nodes with the same
label and the same parent, we can eliminate one. 
Thus, effectively, the maximum number $T_L$ of distinct $L$-level
HSTs satisfies: $T_0 < 2^m$ and $T_L < 2^{T_{L-1}}$ (the reason for the strict inequality is
that we can eliminate the empty tree and also every leaf labeled by the all-ones vector).
The maximum number of leaves $N_L$ of an $L$-level HST therefore satisfies 
$N_L\le \prod_{i=0}^{L-1} T_i$ (in fact, this estimate is far from tight for large $L$).
Notice that $\log_2 N_L = \sum_{i=0}^{L-1} \log_2 T_i = O(\log_2 T_{L-1})$.
Theorem~\ref{thm: 0-infty two-level} can be generalized to any number of levels $L\ge 2$.
Let $c_1 = \lfloor m/2 \rfloor$ and for $L\ge 2$ let $c_L = 2^{c_{L-1}-1}$. Then, for $L\ge 1$,
$$
\Omega(c_L) \le \hat{c}^{\rand}_{L,m} \le O(\log N_L).
$$
The upper bound follows from the above observations and previous work 
(Theorem~\ref{thm: rand upper}). 
As for the lower bound, in the proof of Lemma~\ref{lm: 0-infty two-level rand lower} we 
constructed $2^{2^{m/2-1}}$ one-level subtrees, each labeled with a subset of the
binary cube $\{0,1\}^{m/2}$ of cardinality $2^{m/2-1}$, and also $2^{2^{m/2-1}}$ 
meta-sequences. Each meta-sequence misses one subtree. The subtrees can be
paired into complement halves of the binary cube, and the meta-sequences 
can be paired into complement choices of the sequences that compose them. 
Complement meta-sequences miss complement subtrees. Thus, we can ``lift''
this construction just as we ``lifted'' the uniform metric construction to get a 
$2^{2^{m/2-1}-1}$ randomized lower bound for three-level HSTs, and this ``lifting''
can be iterated ad infinitum. 

For three or more levels, we do not know reasonably tight upper and lower bounds
on the deterministic competitive ratio. We leave this as an open problem. For two 
levels, the upper and lower bounds of Theorem~\ref{thm: 0-infty two-level} are 
similar but not tight asymptotically in $m$. Thus, deriving asymptotically tight bounds, 
and moreover determining if the deterministic and randomized competitive ratios are 
asymptotically (in $m$) equal (as is the case for uniform metric spaces), are interesting 
open problems.

\paragraph{Leaky randomization.}
Generally our results show that the characteristics of the metric spaces have a strong 
influence on the type of guarantees one can hope for, in stark contrast with the well-known 
results and conjectures for non-parametrized sequences (e.g., randomized MTS, or the 
$k$-server problem). Moreover our lower bounds constructions also shed a new light on 
randomization, in the following sense.
In online computing
randomness may help because it hides the state of the algorithm. An adversary
generating the worst-case sequence for a given algorithm knows the probability
distribution but not the outcome of the algorithm's coin flips. A natural question
in this context is to quantify this phenomenon. In particular, consider an adversary
that is given at each step $t$ a signal $\sigma_t$ indicating some information on
the algorithm's state $s_t$ (which is a random variable). Suppose that there exists
$b$ such that for all $t$, the mutual information is $I(\sigma_t; s_t) = b$. If the
adversary is allowed to choose the signals subject to this constraint, what can we
say about the competitive ratio? Before we proceed, two comments are in place.
Firstly, by ``algorithm's state" we could mean simply the position reached by the
algorithm in the state space, or we could mean more broadly also the internal
state of the algorithm. The distinction is immaterial for our results. Secondly,
notice that such an adversary is restricted even from having perfect recall, because
past requests may reveal more than $b$ bits of information regarding the algorithm's
current state. So, we denote the best competitive ratio against the above adversary 
by $c^b_{\cal M}$. In every metric
space ${\cal M}$, the deterministic lower bound in Theorem~\ref{thm: det} implies that
$c^{\lceil\log n\rceil}_{\cal M} = 2n-1$. Notice that $c^{\detr}_{{\cal M},2^b}$ reveals 
something about $c^b_{\cal M}$. In particular, in paired-uniform metric spaces, 
revealing a single bit is sufficient to force the algorithm to pay asymptotically as 
much as a deterministic algorithm.
\begin{corollary}\label{cor: information lb}
If ${\cal M}$ is a paired-uniform metric space, then
$$c^1_{\cal M} = \Omega(n).$$
\end{corollary}

\begin{proofof}{Corollary~\ref{cor: information lb}}
If $c^{\det}_{{\cal M},m} = \Omega(n)$, then in order to force an algorithm
to a competitive ratio of $\Omega(n)$, all that the adversary needs to know
is which of the $m$ requests to use at each step. For this $\log_2 m$ bits
of information are sufficient, and this is clearly an upper bound on the
mutual information. Thus, the corollary follows from Theorem~\ref{thm: paired-uniform}.
\end{proofof}


\bibliographystyle{plain}
\bibliography{biblio}

\end{document}